\documentclass[10pt,journal,final,a4paper,twocolumn]{IEEEtran}
\usepackage{algorithm}
\usepackage{algorithmic}
\usepackage{cite}
\usepackage{amsmath}
\usepackage{amsfonts}
\usepackage{amssymb}
\usepackage{graphicx}
\usepackage{array}
\usepackage{tabularx}
\usepackage{dcolumn}
\usepackage{multicol}
\usepackage{subfigure}
\usepackage{lipsum}
\usepackage{ntheorem}

\begin{document}
%\fontsize{14pt}{15pt}\selectfont %%for english revising
\title{Privacy Protection for Mobile Cloud Data: A Network Coding Approach
%	\thanks{Some part of the results in this paper were published in the International Conference on Parallel Processing (ICPP), 2011.}
	}

\author{
Yu-Jia Chen,~\IEEEmembership{Student Member,~IEEE} and Li-Chun Wang,~\IEEEmembership{Fellow,~IEEE,}
\\National Chiao Tung University, Taiwan
\\Email: allan920693@nctu.edu.tw and lichun@cc.nctu.edu.tw}
%\thanks{National Chiao Tung University, Taiwan. (Email : lichun@cc.nctu.edu.tw and teensky.cm93g@nctu.edu.tw) } }

% The paper headers
%\markboth{Journal on Selected Areas in Communications,~Vol.~x,
%No.~x, May~2010} {Shell \MakeLowercase{\textit{et al.}}: Bare Demo
%of IEEEtran.cls for Journals}

% make the title area
\maketitle

%\vspace{-1.5cm}
\begin{abstract}

Taking into account of both the huge computing power of 
intruders and the untrustedness  of cloud servers, we develop an enhanced secure pseudonym scheme  to protect the privacy of mobile cloud data.
To face the huge computing power challenge, we develop an unconditionally secure light-weight network coding pseudonym scheme.
For the privacy issue of untrusted cloud server, we further design a two-tier network coding to decouple the stored mobile cloud data  from the owner's pseudonyms.
Therefore, our proposed network coding based pseudonym scheme can simultaneously
defend against attackers from both outside and  inside.
We implement our proposed  two-tier light-weight network coding mechanism in a group location based service (LBS) using untrusted cloud database.
Compared to computationally secure Hash-based pseudonym, our proposed scheme is not only unconditionally secure, but also can reduce more than 90\% of processing time as well as  10\% of energy consumption.

\end{abstract}


%\textbf{Note to Practitioners:
%This paper was motivated by the problem of practical privacy protection for Internet of Things (IoT) devices.
%Nowadays, IoT devices can  collect sensitive sensing data, and save their data remotely through the Internet to  cloud servers  for storage  and data analysis, causing privacy concern.
%However, IoT devices cannot implement complex cryptographic scheme because of limited computing power.
%In this paper, we propose a light-weight network coding pseudonym scheme for protecting IoT data privacy.
%We mathematically prove that our scheme is highly secure.
%Experimental results show that the proposed network coding scheme can reduce  the processing time and the energy consumption by more than  90\% and  10\%, respectively.}

\IEEEpeerreviewmaketitle

%\newpage
\section{Introduction}
Data collected from the sensors embedded in smartphones offer great commercial potential for mobile cloud services, but also pose new challenges on privacy protection.
With the properties of continuous changing and updated over time,  mobile data are particularly valuable for big data analytics to understand and predict each individual behaviors. Because mobile data provide highly personal information, the privacy issues of mobile cloud data raise a lot of concerns \cite{reilly2013mobile}.
However, using current security techniques to protect mobile cloud data privacy
will face a number of new challenges in protecting mobile privacy \cite{yu2016big}.
First, we can no longer rely on a protection mechanism with computational security assuming that sophisticated cipher cannot be  broken easily by attackers. This is because the attackers become more powerful in the era of cloud computing.
Second, mobile devices with limited computing power cannot conduct complex encryption and decryption algorithms. Third, private mobile cloud data are stored in public cloud servers, which increases the opportunities to be attacked by other cloud tenants.

To defend against malicious attackers with huge computing power  in outsourced database (ODB), we propose an unconditionally secure network coding based pseudonym scheme to protect mobile privacy against the following two security threats.
First,  for the outsider security threat,   the ownership information of mobile data  is protected to  the highest unconditional security level, in which the ciphertext  will not be broken even with infinite computation time.  Since the inherent distribution nature of network nodes is integrated into network coding, it is impossible to decipher the complete plaintext  unless all the nodes are attacked.
Secondly, for the insider security threat from the attackers staying in the same ODB, we develop a two-tier network coding technique for decoupling data ownership.
As a result, the insider attackers cannot know  the relationship of the stored data and their owners at all.

To demonstrate the advantages of our proposed network coding privacy protection scheme, we implement it for a group location based service (LBS) with ODB \cite{chen2014sudo}.
A group LBS can  help a group of users share their
locations.  Fig.~\ref{f01} shows the proposed security system model for group LBS.
The privacy issues of providing such a group LBS can be classified into the following three aspects:
\begin{itemize}
\item Authenticating users' data while preventing
an untrusted ODB provider from knowing the relation of a user and his/her stored data \cite{khan2012anonymouscloud};
\item Sharing data  with multiple members in the group of interest while maintaining
the same protection level as the two user case  \cite{Lou_03};
\item Storing data in a  multi-tenant  cloud computing environment with the threat of attackers
staying in the same ODB \cite{zhu_03}.
\end{itemize}

%Hence,  the network coding pseudonym  with unconditional security is superior to  the Hash-based pseudonym with computational security.

\begin{figure*}
\begin{minipage}{1\linewidth}
\centering
\includegraphics[width=0.9\textwidth]{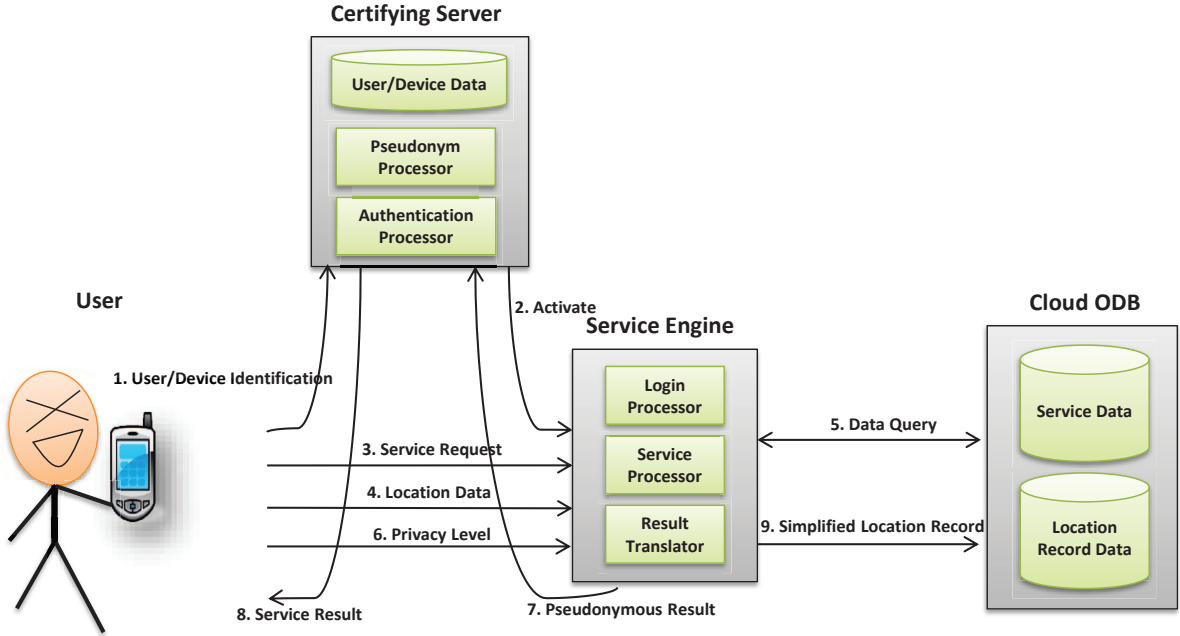}\\
\caption{The proposed security model for a LBS system with  outsourced database (ODB).}
\label{f01}
\end{minipage}
\hfill
\end{figure*}

%In the literature, the privacy preservation issues for LBS  have been  studied extensively\cite{Beresford_09, singelee_2006, del2011privacy, zhu2013toward,abbas2013privacy}. These works can be classified into two types : (1) location tracking group LBS and (2) position-aware group LBS.  The location tracking LBS can frequently track users' locations, such as Google Latitude, while the position-aware LBS, such as Facebook's Places,  provide the locations of neighboring stores.
%Although some works have been proposed for protecting privacy in support of LBS \cite{abbas2013privacy,zhu_03},

%The main idea of the proposed security scheme is to generate network coding pseudonym such that the relation of data and its identity is decoupled, which is diffident from the existing Hash-based pseudonym.

An unconditionally secure  framework for mobile cloud applications has been rarely seen in the literature. Our proposed unconditionally secure network coding based pseudonym scheme
can overcome the aforementioned three challenges, and has the following two major contributions:
\begin{itemize}
\item A data ownership decoupling mechanism is presented, which
can  achieve unconditional security level and resolve the security concerns for
untrusted ODB.
\item  A light-weight network coding  scheme is proposed  which can reduce the processing time by 90\% compared to the Hash-based pseudonym.  A user only needs to encode his/her identity one time when logging in to the trusted certifying server.
The remaining operations for  pseudonym changing and uncloaking are performed by the trusted certifying server. None of decoding operations are needed in mobile terminals.

\end{itemize}

The rest of the paper is organized as follows. Section II discusses the related work for location privacy protection. In Section III, we explain the system security model for group LBS using ODB. In Sections IV and V, we present the proposed network coding based pseudonym and analyze its security performance. In Section VI, we discuss the issues of data ownership, user authentication, and service continuity. Section VII compares the performance of the proposed network coding based pseudonym with Hash-based pseudonym. Finally, concluding remarks are given in Section VIII.

\section{Related Works}
In the literature,  location privacy is generally protected  by either pseudonyms or anonymization.

\subsection{Pseudonym-based Location Privacy Protection}

Pseudonym techniques can protect  location privacy by disconnecting  a user's location data with his/her genuine identity \cite{Beresford_09, singelee_2006, del2011privacy, zhu2013toward, liu2013game, lu2012pseudonym, mathews2014effective}.
Existing pseudonym schemes can be classified into two types.
The first type of pseudonym techniques aim to change the pseudonyms intelligently so that the attacker cannot distinguish the spatial difference from other members.
The authors of \cite{Beresford_09} proposed to frequently change pseudonyms based on the mixed-zone concept, where users' pseudonyms are mixed together.
In \cite{zhu2013toward}, disposable interface identifiers were frequently assigned to users.
When a node is not allowed to disclose addresses, a silence period (defined as a transition period between changing pseudonyms) was  introduced \cite{lu2012pseudonym}.
When users' pseudonyms in some area are changed after a silence period, a hacker could not find the target user.  In \cite{liu2013game}, the pseudonym mechanism was improved by using a Bayesian approach to resolve the frequency changing issue. In \cite{freudiger2013non}, a  game-theoretic model for the non-cooperative behavior analysis of  mobile nodes in a mixed zone was proposed. A location proof updating system for vehicular networks was designed in \cite{zhu2013toward}.  In \cite{ying2013dynamic}, the authors proposed a dynamic mixed zone of which the size is determined by the vehicle's predicted location, traffic statistics, and privacy requirement.

 % Hash
Another type of pseudonym techniques focus on designing the unlinkable and irreversible pseudonym based on cryptographic method. The Hash function is widely used to generate a unique and secure pseudonym \cite{mathews2014effective}.
However, it has been shown that attackers with computational power provided by cloud computing can break Hash-based cryptography \cite{teat_11}.
The brute force method can deduce the original message from a Hash value \cite{kumar_2013}.
Furthermore, existing pseudonym techniques cannot completely protect data ownership if an ODB provider can link a particular pseudonym with its stored data in ODB during an authentication process.

\subsection{Anonymization-based Location Privacy Protection}
Anonymization techniques can protect location privacy by generalization and suppression techniques, i.e., with an expression of  lower granularity. In \cite{Gruteser_04} the $k$-anonymity concept was introduced so that a user's location cannot be distinguished from other $k-1$  users' locations. Instead of reporting the exact location, a user sends a region containing the locations of other $k-1$ people. In \cite{Gedik_17}, the CliqueCloak algorithm was proposed to provide different $k$-anonymous requirements for each user.  A privacy-preserving architecture for LBS with different anonymization techniques was suggested in \cite{Yin_19}. A node density-based location privacy scheme was proposed to improve anonymity without degrading quality of service (QoS) \cite{miura2013hybrid}. The authors of  \cite{gurjar2013cluster} proposed  a cluster based anonymization scheme to replace the real node identities with random identities generated by the cluster heads.
In \cite{ying2014protecting}, a clustering anonymization for vehicular networks was developed to hide the road and traffic information. However, all the above anonymization techniques cannot protect location data shared by a group of users since an individual's location is concealed.

\subsection{Objective of This Paper}

The objective of this paper is to improve location privacy protection to unconditional security level and prevent an ODB provider from linking a particular pseudonym with its stored data.
To our knowledge, privacy protection for group LBS in the context of ODB
has not been fully investigated. We suggest to adopt network coding for pseudonym  generation, which  has received little attention in the literature so far.

Network coding is a generalized routing method in which messages are computed by intermediate nodes with algebraic encoding.
Previous works have shown that network coding can provide robustness \cite{koetter2003algebraic, 6134039, cai2002network} and improve throughput \cite{6134039,zeng2014throughput} as well as confidentiality \cite{vilela2008lightweight}.
More recently, network coding is employed to improve the security and reliability of distributed storage systems such as recovering lost data in multiple storage nodes \cite{chen2014nccloud}, preventing eavesdropping over untrusted networks \cite{cheneavesdropping} and checking integrity of outsourced data \cite{chen2014enabling}.

Although it was proved that network coding possesses the properties of irreversibility  \cite{oliveira2012coding} and verifiable data integrity \cite{chen2016secure}, a stronger security property of network coding that can achieve unlinkability between coded data has been rarely studied yet. In our previous work \cite{Chen_11}, the network coding based pseudonym scheme and the corresponding security mechanism for group LBS using ODB was introduced, but the security analysis of the proposed mechanism is not yet investigated.

To this end, we develop a two-tier coding method to mix user identity with watchword/seed and then generate two keys for certification and anonymization.
In this paper,  on top of the work in  \cite{Chen_11}, we further analyze the security performance of the proposed  pseudonym scheme.
We will prove that network coding can provide unlinkable pseudonyms to unconditional security level for group LBS by analysis.
It will be demonstrated that the proposed two-tier network coding method can
preserve the privacy of data ownership even if an untrusted ODB provider has huge computing power.

\section{System Model} \label{System_model}

\subsection{Location Based Service (LBS) Applications Scenario}

\begin{figure}
\centering
\includegraphics[width=0.5\textwidth]{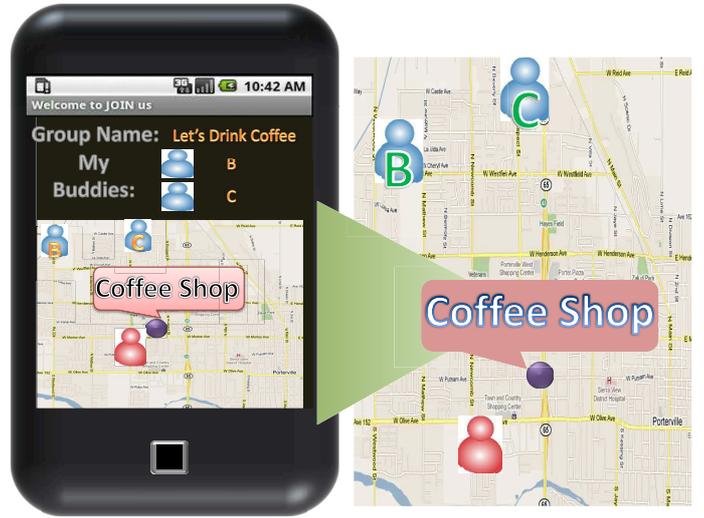}\\
\caption{ An illustrative example of a group LBS for scheduling a meetup.} \label{f02}
\end{figure}
As mentioned before, the group LBS is considered as an example to investigate the security performance issues for mobile cloud applications.
With users' interests as well as their locations, LBS system  establishes a ``group" of users with the same interests to help customers' meetup activities.
We  implemented a location-based group scheduling service on the ODB model, called JOIN \cite{Lee_05}, which can easily  arrange a meetup  by
sharing their current  available time and locations for polling, voting, and broadcasting.
Fig. \ref{f02} illustrates a service scenario of JOIN, in  which
users A, B and C  join the  ``Let's drink coffee" group.
While user A is approaching to his preferred coffee shop, users B and C   are nearby coincidentally.
 JOIN server  will provide the proximity information of  its group members users B and C to customer A, and the advertisements for  this coffee shop.
The locations and activities of customers in this meetup will be stored in the cloud database for future commercial promotions, such as group coupon distribution.

In this group LBS  application scenario, it is worthwhile investigating the following security  issues. First,  customers shall release  their identities to the service provider and to other group members. During the sharing process, it is likely that a hacker will eavesdrop the location data and their identities.  Moreover, due to the use of ODB,  an adversary can possibly pretend to be an authorized user and then modify the customers' private information. Clearly,  a more trustworthy LBS provider needs to protect  customers'  identities and location data in a more
sophisticated security mechanism.

\subsection{Security Models and Assumptions}

Figure \ref{f01}  shows the proposed security model for a mobile cloud system with ODB,
consisting of users devices,  a service engine, cloud databases, and a trusted certifying  server.
The service engine  contains  three major processors: the login processor, service data processor (e.g., scheduling a meetup), and service result translator (e.g., anonymization).
The cloud database stores the service  information  and location data.
Rather than in  cloud databases, customers' data are stored at the certifying servers, consisting of authentication processor and pseudonym processor.
Referring to Fig. \ref{f01}, the general information flows of privacy preserved LBS are described as follows.
\begin{enumerate}
\item  Step 1: A customer sends his/her identification as well as the device identification to the certifying server. Then the authentication processor within the certifying server  verifies customers' identities.

\item Steps 2 and 3: The login processor is activated to respond the service requests of customers.

\item Steps 4 and 5: The customer sends the location data to the service engine. Then the service processor can access other service-related information from the cloud database, such as the addresses of the nearby coffee shops.

\item Steps 6 and 7: After receiving the privacy level requested by the customer, the result translator sends the pseudonym to the certificate server.

\item Step 8: The pseudonym processor  ``uncloaks" the pseudonymous result.

\item Step 9: The simplified location records are stored in the cloud database with various privacy levels.
A low privacy level could show group members at the scale of street locations, whereas the high privacy level could  show your friends only at the scale of a city.
\end{enumerate}

 To make the privacy  analysis more tractable, the considered JOIN group LBS is implemented in an experimental platform, invoking the following assumptions and conditions:
 \begin{itemize}
 \item  A group of users trust each other so as to being willing to share their locations and identities.
 \item  The service provider and the cloud database provider are semi-honest \cite{chai2012verifiable}.  Hence, these providers shall follow the protocol properly, while allowing them to  keep all the records of  involved computations and data exchanges. \end{itemize}
Noteworthily, under the semi-honest requirement, there still exists a possibility  of deriving private data of other parties from the intermediate computation records.

%The ODB can allow customers to utilize extremely large storage capacity through the Internet, but it raises the security issue of protecting mobile users' locations data.

%The goal of this paper is to develop pseudonym technique for group LBS to protect the privacy of the location data ownership information in mobile cloud computing environment.

\section{IMSI-Based Group Security (IGS) Algorithm}\label{IMSI-BASED GROUP SECURITY (IGS) ALGORITHM}

\subsection{Algorithm Design}

In this subsection, we  present network coding based pseudonym integrated with the international mobile subscriber identity (IMSI).
Stored in a subscriber identity module (SIM) card, IMSI  is a unique number associated with  current mobile phones to identify the subscribers.
Current Android-based mobile platform provides the application programming interface (API) to access the IMSI number of mobile phones.
Clearly, location data themselves are not quite useful as long as  a hacker does not know the
ownership of these data.
Hence, we propose to encrypt IMSI  to conceal a customer's privacy.

%Then we further detail the IMSI-based group security (IGS) algorithm and the corresponding service procedure when applying the proposed mechanism in the JOIN services.

Fig.~\ref{3} (a) and (b) illustrate the basic idea of the proposed IMSI-based group security (IGS) algorithm, in which a two-tier coding scheme is developed to generate KeyA  for authenticating a legal customer's identity,  and KeyB (pseudonym) for protecting customer's private data (e.g., location).
First, KeyA is generated by mixing customers' IMSI and customers' watchwords,
which are  selected by customers as their secrets to protect their IMSI information.
Watchwords shall not be revealed to any party.
Second,  mixing each customer's KeyA and a uniformly distributed random seed results in
KeyB. Note that the length of watchwords and random seed should be
longer  than or equal to the length of the IMSI.
The key generation functions  in the proposed two-tier coding scheme
shall satisfy the following two requirements even if the key generation function is revealed:	(1) difficult to obtain any information from KeyA or KeyB; (2) difficult to obtain the correspondence between KeyA and KeyB. We consider two  kinds of  key generation functions in this paper.
\begin{itemize}
\item  {\it Hash function}:   It has been used extensively
in the applications of  message integrity checking and authentication, such as message digest (MD5) and secure Hash algorithm (SHA).

\item   {\it Network coding}: In efficient routing  and   secure storage schemes \cite{oliveira2012coding}, network coding technique has been adopted.
Vandermonde transform matrix is the key element in network coding.
\end{itemize}
Vandermonde matrix is used in our proposed  two-tier network coding scheme.
Denote $\mathbf{A}$ as an n$\times$n Vandermonde matrix, where $[\mathbf{A}_i,_j]=(a_j^i-1)$ and  the coefficients $a_i$ are distinct nonzero elements in a finite field $F_q$ , $q=2^u>n$.  Let $\mathbf{b}=(b_1,\dots,b_n)^T$ be the mix of IMSI and watchwords.  Next, we compute $\mathbf{c}=(c_1,\dots,c_n)^T=\mathbf{A}\mathbf{b}$ and randomly select the segment of $\mathbf{c}$  as the key value in the two-tier network coding scheme.

%%%%%%Nov. 7

The aforementioned key generation functions can be adopted in our pseudonym generation scheme.
Fig.~\ref{3} shows the proposed two-tier coding
because a pseudonym in our scheme is generated from two sequential processes.
In the first tier,  IMSI and a watchword are mixed by the key generation function to yield a cipher pseudonym.
This can protect the data ownership privacy during an authentication process.
In the second tier, the cipher pseudonym of the first tier is mixed with random seeds
to result in the second cipher pseudonym.
We use the second cipher pseudonym to protect the data privacy
when sharing data with others and  storing data in an ODB.

\begin{figure}
\centering
\includegraphics[width=0.5\textwidth]{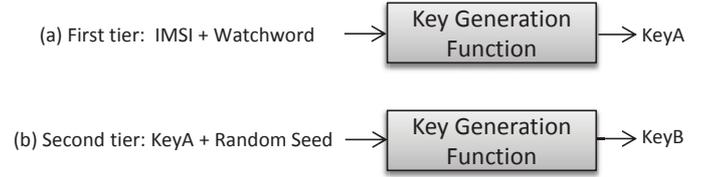}\\
\caption{ Proposed two-tier coding scheme with encrypted international mobile subscriber identity (IMSI), where KeyA and KeyB are used for authentication and  pseudonym generation, respectively. } \label{3}
\end{figure}

Now we detail the proposed IGS algorithm.
Initially, a customer sends his/her genuine identity and KeyA to the certifying server.
If KeyA is certificated, the certifying server generates an initial pseudonym (KeyB) by the input KeyA, and then activates the login processor.
Algorithm $1$ describes the detail procedures of the proposed pseudonym generation algorithm. Note that  the tolerance distance  and the silence period reflect the customers' preferred privacy level and QoS requirements.
The former  indicate the maximum acceptable location bias of a customer, and the later specifies  the  period of  changing pseudonyms \cite{Leping_15}.
The pseudonym processor exchanges KeyBs under the condition that  the pseudonym timer reaches the silence period and any other group member is within the range of the tolerance distance. Otherwise, a new KeyB will be regenerated.
Since the pseudonyms are mixed with  nearby members' pseudonyms,  the proposed IGS algorithm will impose ambiguity on  a customer's exact location and thus
protect his/her location privacy. Finally, a customer uses the generated KeyB  to update the current location.

\begin{algorithm}
\caption{Pseudonym (KeyB) generation algorithm}
\begin{algorithmic}[1]
\IF  {a customer sends the KeyA which is authorized by the certifying server}
           \STATE activate the login processor
           \STATE generate and send KeyB to the customer
           \STATE start a timer
\ELSE
           \STATE send a reject information to the customer

\ENDIF

\IF  {timer $\geq$ silence period}

           \IF {tolerance distance $\geq$ distance(customer, a randomly selected friend F within the same group) }
                \STATE exchange KeyB of the customer and F and also reset the timer
           \ELSE
                \STATE send a command of KeyB's changing to the customer and also reset the timer
           \ENDIF
\ENDIF

  \label{code:recentEnd}
  \end{algorithmic}
\end{algorithm}

\subsection{Proposed IMSI-based Group Secure (IGS) Algorithm }

In this subsection we detail the IGS algorithm and the  service procedures
when applying the proposed pseudonym generation algorithm in JOIN.
Fig.~\ref{4}(a) shows the register process.
Each customer should register a unique account.
Initially, a customer generates KeyA with his/her watchwords and transmits KeyA to the  server with password and genuine identity.
Fig.~\ref{4}(b) shows the login process of the JOIN services.
A process thread is initiated and retained after the verification of the identification and the password.
Since the thread is retained, the certifying server can identify the customer in the communication afterwards.
Furthermore, the certifying server generates an initial KeyB from KeyA and a random seed.
Then the KeyB is delivered to the customer.

Fig.~\ref{5}(a) shows the activity initiation process of the JOIN services,
where the dotted line and the solid line in this figure is used to
denote data stored in the memory and on the hard disk, respectively.
A customer can initiate a new activity (i.e., invites friends to go somewhere) and receive information of nearby friends by sending KeyB and location to the  server.
Then the server sends the requests to all other group members.
Each certificated member has to send his/her KeyB and locations to the server responding to this request.
Next the JOIN server searches the nearby friends from the location record table and provides some group information (e.g., the top 10 restaurants, coffee shops, or bakeries with fresh-baked bread) from the service data in the cloud-based ODB.
Then this group information as well as KeyB and the location of the nearby friends are sent to the certifying server.
The certifying server decodes KeyB with a genuine identity and generates a new KeyB based on the proposed IGS algorithm.
Finally, the customer receives the service result from the certifying server.

Fig.~\ref{5}(b) shows the storage process of the JOIN services.
The JOIN server stores each customer's KeyB and location in the cloud database.
Note that storing a less accurate location in the ODB can result in a higher privacy protection.
For example, storing New York City is more secure than storing Manhattan in the database.
Finally,  data in the memory will be cleared after the activity is finished.
\begin{figure*}
\centering
\includegraphics[width=0.6\textwidth]{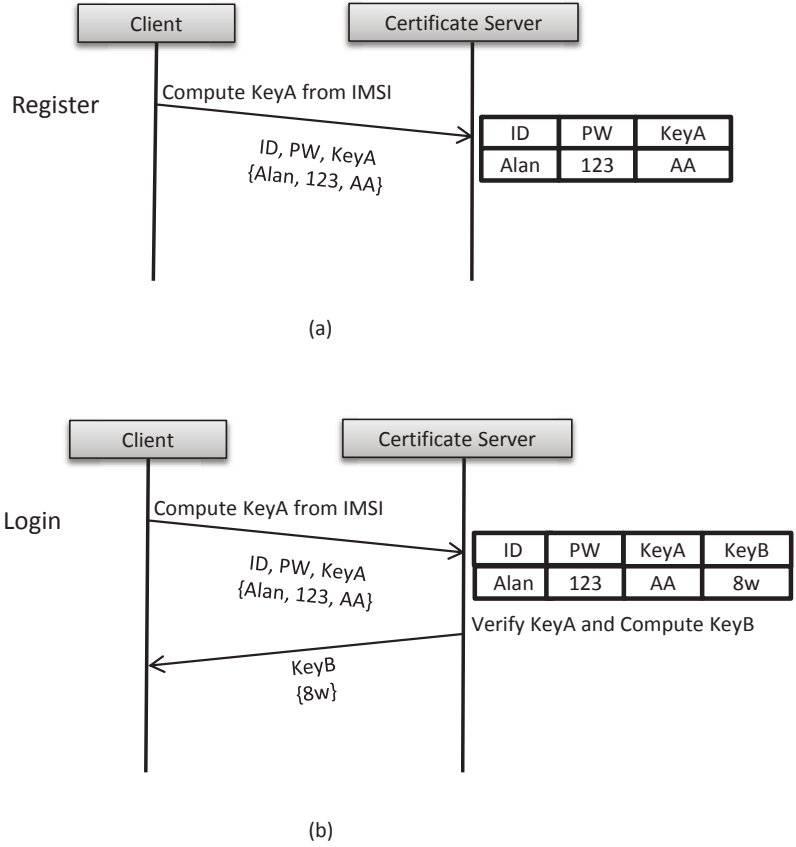}\\
\caption{Register and login process of the JOIN services.} \label{4}
\end{figure*}

\begin{figure*}
\centering
\includegraphics[width=0.8\textwidth]{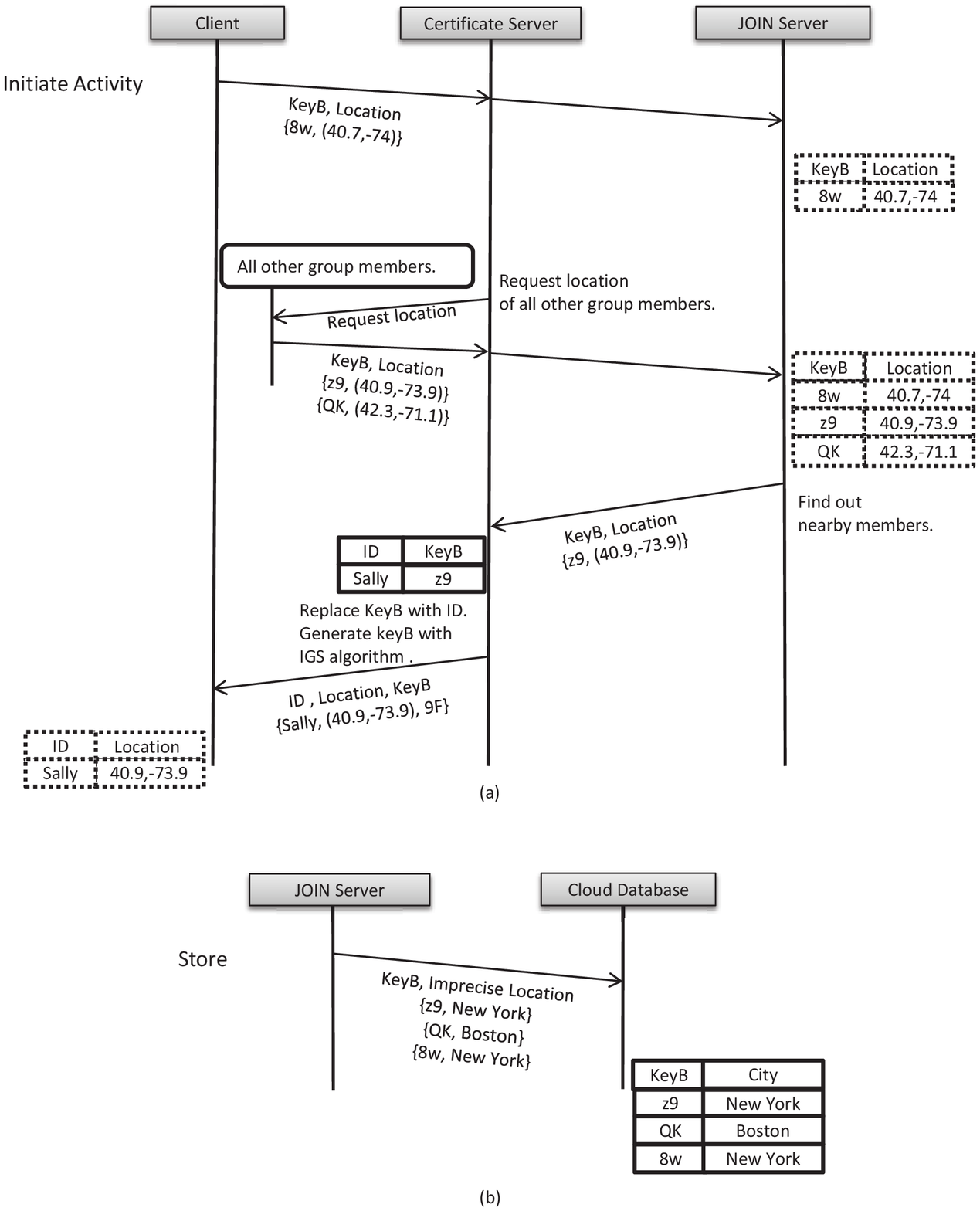}\\
\caption{Activity initiation and storage process of the JOIN services.} \label{5}
\end{figure*}

\section{Privacy Analysis}\label{SECURITY ANALYSIS}

\subsection{Security Problem Formulation}

In this subsection, we  will prove that the proposed network coding based pseudonym scheme can guarantee the unlinkability of locations and customers to the unconditional security level.
Although   the
irreversibility feature of network coding   can  prevent a hacker from deducing valuable information \cite{oliveira2012coding}, a hacker can still link the relation between the pseudonyms of the same customer.
Thus,  the unlinkability feature is desirable for a pseudonym scheme. First, the hardness assumptions for Hash functions are given. Then we present the information-theoretic security proof for the proposed  network coding pseudonym scheme. The security of  the  group LBS security model
is evaluated based on
the probability of a hacker being able to break into a security system.
The basic idea of our security system is to
use a customer's pseudonym to hide his/her identity.
The pseudonyms are  generated by encrypting the IMSI of the customer's device  through
either network coding or a Hash function.

% Nov. 18.
Consider a scenario where adversary A can eavesdrop all the links on the Internet, and
know the detailed information of key generation functions (i.e., Hash function or Vandermode matrix $\mathbf{A}$). The goal of this adversary  is to find the correspondence between the customer's ID and his/her locations.
Since a customer's (ID, KeyA) and (KeyB, location) are stored in two separate storage nodes,  adversary A can break the system once  the relations of
(ID, KeyA) and (KeyB, location) are known.
A security problem for this kind of group LBS can be formulated as follows.

% Nov. 23.

\newtheorem*{theorem-non}{IMSI Security Problem}
\begin{theorem-non} Given $k({A_{i}})$ and $k({B_{j}})$, determine whether $IMSI_{i}=IMSI_{j}$.
Here $k({A_{i}})$ and $k({B_{j}})$  represent the $i$-th customer's keyA and the $j$-th customer's keyB generated from the same key generation function $k()$, respectively. We denote the mixing of IMSI and watchwords as ${A_{i}}$ and the mixing of KeyA and random seeds as ${B_{i}}$.
\end{theorem-non}
Next, we derive the following theorem.
\newtheorem{theorem}{Theorem}
\begin{theorem}
Adversary A can find the correspondence between ID and location if and only if A can correctly solve $IMSI$ $Security$ $Problem$.
\end{theorem}
\begin{proof}
Assume that adversary A can find the correspondence between ID and location.  Given the location associated with  $k({B_{j}})$, adversary A can know its corresponding  ID of the customer. By eavesdropping the link between the device and the certifying server, adversary A can get the customer's $k({A_{j}})$.
If $k({A_{j}})=k({A_{i}})$, the answer for  $IMSI$ $Security$ $Problem$ is ``YES"; otherwise, the answer is ``NO". By contrast, assume an adversary A can solve $IMSI$ $Security$ $Problem$. For any given $k({B_{j}})$ and its corresponding location $l$, adversary A can answer $IMSI$ $Security$ $Problem$ for each KeyA in the database records.  If  the answer is ``YES"  for a certain KeyA with ID $d$, then $d$ is located at $l$.
\end{proof}

To reduce the possibility of $IMSI$ $Security$ $Problem$ to be solved, the hardness assumption of the Hash function used as the key generation function is given as follows: It is computationally infeasible in determining whether two output Hash values lead to the fact that  part of two corresponding input messages are the same.

\subsection{Analysis of Network Coding based Pseudonym}
Now we evaluate our security framework for which network coding techniques are used to encrypt the customer's IMSI.
Referring to $Theorem$ $1$, the perfect privacy protection is to ensure the
probability of $IMSI$ $Security$ $Problem$ being solved is equal to  the probability of random guessing.
Next, we give the definition that the considered LBS system is $secure$.
\newtheorem{definition}{Definition}
\begin{definition}
The system is $secure$ if $|\Pr\{$An adversary correctly solves $IMSI$ $Security$ $Problem$$\} - \frac{1}{2}|\ \leq  \varepsilon$, where $\varepsilon$ is $negligible$. Note that a function $f:\mathbb{N} \rightarrow \mathbb{R} $ is called $negligible$ if for every polynomial $p$ there exists a positive integer $N(p)$ such that $|f(n) \leq \frac{1}{p(n)}|$ for all $n \geq N(p)$.
\end{definition}
\newtheorem{lemma}{Lemma}
\begin{lemma}
Assume an LBS system generate $\{k({A_{i}}),k({B_{j}})\}$.
Let $A[k({A_{i}}),k({B_{j}})]$ be the guess of adversary A. If
\begin{equation}
Pr\{A[k({A_{i}}),k({B_{j}})] = Yes|IMSI_{i}=IMSI_{j}\} \leq \frac{1}{2} + \varepsilon_{1}
\enspace \label{lemma1.1}
\end{equation}
and
\begin{equation}
Pr\{A[k({A_{i}}),k({B_{j}})] = No|IMSI_{i}\neq IMSI_{j}\} \leq \frac{1}{2} + \varepsilon_{2}
\enspace, \label{lemma1.2}
\end{equation}
where $\varepsilon_{1}$ and $\varepsilon_{2}$ are $negligible$, then the system is $secure$.
\end{lemma}
\begin{proof}
Let
\begin{equation}
P_{Y} \triangleq Pr\{A[k({A_{i}}),k({B_{j}})] = Yes|IMSI_{i} = IMSI_{j}\}
\end{equation}
and
\begin{equation}
P_{N} \triangleq \\ Pr\{A[k({A_{i}}),k({B_{j}})] = No|IMSI_{i}\neq IMSI_{j}\} \enspace.
\end{equation}
Without loss of generality, assume an LBS system with only two customers.
We have that
$Pr\{$An adversary correctly solves $IMSI$ $Security$ $Problem$$\} = \frac{1}{2}P_{Y} + \frac{1}{2}P_{N}$
$\leq \frac{1}{4} + \frac{1}{2} \varepsilon_{1} + \frac{1}{4} + \frac{1}{2} \varepsilon_{2}$.
Thus, it follows that $|Pr\{$An adversary correctly solves $IMSI$ $Security$ $Problem$$\} - \frac{1}{2}| \leq  \varepsilon$,
where $\varepsilon = \frac{1}{2}(\varepsilon_{1} + \varepsilon_{2})$ is $negligible$.
\end{proof}

Next, we analyze the independence of the IMSI and the generated pseudonyms.
\begin{theorem}
The information mutually held between IMSI and KeyA or KeyB of a customer is
$I(k({A_{i}});IMSI_{i}) = I(k({B_{j}});IMSI_{j}) = 0$, where the mutual information $I()$ is defined by $I(X;Y) = \sum\limits_{y \in Y} {\sum\limits_{x \in X} {p(x,y)} } \log \left( {\frac{{p(x,y)}}{{p(x)p(y)}}} \right)$.
\end{theorem}
\begin{proof}
Let $\mathbf{e}^{(h)}$ represent a subset containing arbitrary $h$ components of vector $\mathbf{e}$. Denote $\mathbf{e}_{i:j}$ as the subvector formed from the $i$-th to the $j$-th position of vector $\mathbf{e}$. The set of the $i$-th and the $j$-th rows of matrix $\mathbf{D}$ is represented as $\mathbf{D}_{i:j}$. Assume that the length of IMSI be $m$ and the length KeyA and KeyB be $k$.
Represent $\mathbf{b}$ as the uncoded data (e.g., the mixing of IMSI and watchwords), where $b_{i}$ are uniformly distributed independent random variables over $\mathbf{F_{\emph{q}}}$ with entropy $H(b_{i}) = H(b)$, and $H()$ is defined by $H(X) =  - \sum\limits_{i = 1}^n {p({x_i})\log } p({x_i})$.
For simplicity,  the key is selected from arbitrary contiguous $k$ components of $\mathbf{c} = (c_{1},\cdots,c_{n})^{T} = \mathbf{Ab}$.
The mutual information $I(\mathbf{c}_{p+1:p+k};\mathbf{b}^{(m)})$ for $0 \leq p \leq n-k$ can be calculated as:
\begin{eqnarray}
&& I(\mathbf{c}_{p+1:p+k};\mathbf{b}^{(m)}) \nonumber \\
&&= I(\mathbf{b}^{(m)};\mathbf{c}_{p+1:p+k}) \nonumber \\
&&= H(\mathbf{b}^{(m)}) - H(\mathbf{b}^{(m)}|\mathbf{c}_{p+1:p+k})
\enspace, \label{entropy}
\end{eqnarray}
where
\begin{align}
& H(\mathbf{b}^{(m)}|\mathbf{c}_{p+1:p+k})  & \nonumber \\
&= \sum_{j=1}^m H(b_{seq(j)}|   \mathbf{c}_{p+1:p+k}, & \nonumber \\
&\hspace{30pt} b_{seq(j-1)} = \check{b}_{seq(j-1)},\cdots,b_{seq(1)} = \check{b}_{seq(1)} )
\enspace, \label{extend_entropy}
\end{align}
\begin{figure*}
\begin{align}
{\mathbf{M}}_{{\text{p + 1:p + k}}}  = \left[ {\begin{array}{*{20}c}
   {m_1^p } & {...} & {m_p^p } & 1 & {...} & 0 & {m_{p + k + 1}^p } & {...} & {m_n^p }  \\
    \vdots  &  \vdots  &  \vdots  & 0 &  \ddots  &  \vdots  &  \vdots  &  \vdots  &  \vdots   \\
   {m_1^{p + k - 1} } & {...} & {m_p^{p + k - 1} } & 0 & {...} & 1 & {m_{p + k + 1}^{p + k - 1} } & {...} & {m_n^{p + k - 1} }  \\
 \end{array} } \right]
\enspace \label{bigmatrix}
\end{align}
\end{figure*}
In (\ref{extend_entropy}) $seq(j)$ is the $j$-th element of a random integer sequence within the range $1$ to $n$, and $\check{b}_{i}$ represents the given value for the random variable $b_{i}$.
Since the $n\times n$ Vendermonde matrix $\mathbf{A}$ is nonsingular \cite{Klinger_21}, we can apply the Gaussian elimination to obtain the reduced row echelon form of the submatrix $\mathbf{S}$, in $[S_{i,j}] = [A_{i,j}]$, $p+1 \leq i,j \leq p+k$.
Then \textit{Adversary Reduced Matrix} $\mathbf{M}$ is obtained as (\ref{bigmatrix}), 
where the other element of $\mathbf{M}$ are the same as $\mathbf{A}$.
We can also obtain \textit{Adversary Reduced vector} $\mathbf{v}$ of which the elements
$v_{p + 1} \cdots v_{p + k}$ are represented as
\begin{align}
{\mathbf{v}}_{{\text{p + 1:p + k}}}  = \left[ {\begin{array}{*{20}c}
   {v_{p + 1} }  \\
    \vdots   \\
   {v_{p + k} }  \\
 \end{array} } \right] \enspace.
\end{align}
Each element in $\mathbf{M}$ and $\mathbf{v}$ results from the basic row operations to reduce $\mathbf{S}$ to $\mathbf{I_{k}}$. Thus, we obtain $k$ equations to solve $n$ unknown elements $b_{i}$:
\begin{align}
{{v}}_{{i}} = \sum\limits_{j = 1}^p {m_j^{i - 1} b_j }+ b_i+\sum\limits_{j = p + k + 1}^n {m_j^{i - 1} b_j } \enspace,
\end{align}
where $p + 1 \leq i \leq p + k $.
%\enspace, \label{k_equations}
%\end{eqnarray}
Since we cannot solve any $b_{i}$ without  $n-k$ components of $\mathbf{b}$ for $1\leq j\leq n-k$, it follows that
\begin{align}
&H(b_{seq(j)}|\mathbf{c}_{p+1:p+k},b_{seq(j-1)} = \check{b}_{seq(j-1)},\cdots,b_{seq(1)} = \check{b}_{seq(1)} ) \nonumber \\
&= H(b_{seq(j)}) \nonumber \\
&=H(b)
\enspace. \label{H_equation_equal}
\end{align}
For $n-k\leq j\leq m$, the number of equations is more than the number of unknown elements.
Thus, we obtain
\begin{align}
&H(b_{seq(j)}|\mathbf{c}_{p+1:p+k},b_{seq(j-1)} = \check{b}_{seq(j-1)},\cdots,b_{seq(1)} = \check{b}_{seq(1)} ) \nonumber \\
&=0
\enspace. \label{H_equation_equal_0}
\end{align}
Substituting (\ref{H_equation_equal}) and (\ref{H_equation_equal_0}) into (\ref{extend_entropy}), we have
\begin{eqnarray}
H({{\mathbf{b}^{\left( m \right)}}|{\mathbf{c}_{p + 1:p + k}}})\  = \left\{\begin{array}{ll}
mH(b)\enspace , & \mbox{$m\leq n-k$} \\
(n - k)H\left( b \right) \enspace ,& \mbox{$m>n-k$}~
\end{array} \right.. \label{both equation}
\end{eqnarray}
Since $b_{i}$ are i.i.d random variables, it follows that
\begin{eqnarray}
H({{\mathbf{b}^{\left( m \right)}}})&=& H\left( {{b_{seq(1)}},{b_{seq(2)}}, \ldots ,{b_{seq(m)}}} \right) \nonumber \\
&=& mH(b)  \enspace. \label{right equation}
\end{eqnarray}
Finally, substituting (\ref{both equation}) and (\ref{right equation}) into (\ref{entropy}), we can obtain
\begin{equation}
I\left( {{\mathbf{c}_{p + 1:p + k}};{\mathbf{b}^{\left( m \right)}}} \right) = \left\{\begin{array}{ll}
0 \enspace , & \mbox{$m\leq n-k$} \\
(m - n + k)H\left( b \right) \enspace ,& \mbox{$m>n-k$}~
\end{array} \right..
\end{equation}
Note that we select $k=n/2$ and $m\leq n/2$ in our network coding based pseudonym scheme to ensure
$I\left( {{\mathbf{c}_{p + 1:p + k}};{\mathbf{b}^{\left( m \right)}}} \right) = 0$.
\end{proof}

According to $Theorem$ $2$, we know that KeyA or KeyB are independent of IMSI.
As a result, the probability for an adversary to correctly solve $IMSI$ $Security$ $Problem$ is the same as random guessing. This is implied that
\begin{eqnarray}
&&{\text{Pr}}\left\{ {{\text{A}}\left[ {{\text{k}}\left( {{{{A_i}}}} \right),{\text{k}}\left( {{{{B_j}}}} \right)} \right] = Yes | {\text{IMS}}{{\text{I}}_i} = {\text{IMS}}{{\text{I}}_j}} \right\}  \nonumber \\
&=& {\text{Pr}}\left\{ {{\text{A}}\left[ {{\text{k}}\left({{{{A_i}}}} \right),{\text{k}}\left( {{{{B_j}}}} \right)} \right] = Yes} \right\} \nonumber \\
&=& \frac{1}{2} \nonumber \\
&\leq& \frac{1}{2} + {\varepsilon _1}
\end{eqnarray}
and
\begin{eqnarray}
&&{\text{Pr}}\left\{ {{\text{A}}\left[ {{\text{k}}\left( {{{{A_i}}}} \right),{\text{k}}\left( {{{{B_j}}}} \right)} \right] = No|{\text{IMS}}{{\text{I}}_i} \ne {\text{IMS}}{{\text{I}}_j}} \right\} \nonumber \\
&=& {\text{Pr}}\left\{ {{\text{A}}\left[ {{\text{k}}\left( {{{{A_i}}}} \right),{\text{k}}\left( {{{{B_j}}}} \right)} \right] = No} \right\}  \nonumber \\
&=& \frac{1}{2} \nonumber \\
&\leq& \frac{1}{2} + {\varepsilon_2}  \enspace.
\end{eqnarray}
Based on $Lemma$ $1$, we conclude that the proposed network coding scheme can achieve the unconditional security level.

\section{Discussion}
In this section, we discuss the security issues in terms of privacy, authentication, and continuity for the proposed group LBS scheme. Then, we compare the security performance of the proposed pseudonym scheme with that of traditional pseudonym schemes.

\subsection{Data Ownership Privacy}

As shown in the previous section, network coding can provide unconditional security rather than computational security provided by Hash functions.
Note that the unconditional security property of the proposed IMSI pseudonym is contributed from the uncertainty of the customer's watchword.
Therefore, the privacy of location data ownership is fully preserved even if one can eavesdrop IMSI, KeyA, or KeyB.
This property ensures that the proposed pseudonym scheme can prevent the ownership information from being accessed by eavesdroppers.
For the same reason, our system is collusion-resistant to the insider attack even when a service provider and a cloud database provider collaboratively attempt to decrypt the pseudonyms.
Furthermore, it is discussed that the connection between a pseudonym and its ownership information can be derived by analyzing messages (e.g., location and group member information) subsequently \cite{narayanan2009anonymizing}.
In the proposed IGS algorithm, customers' keyB is exchanged randomly within the same group after a silence period, thereby increasing the difficulty of de-anonymization.
Compared to the approach using a new KeyB, the proposed method can confuse the eavesdropper because the customers in the same group have similar interests and movement patterns.
The silence period is a design parameter specified by a random variable within a certain range \cite{Leping_15}.
When the durations of two customers' KeyBs are overlapped, the temporal relation of these two customers cannot be correctly linked from the new KeyB and the old KeyB.
In general, a longer silence period results in lower location timeliness (i.e., the elapsed time since the location was acquired) but provides higher privacy protection.
Thus, a cloud database cannot store the precise locations because they are strongly associated with customers.
For example, if a customer usually stays in one place, we can guess that he/she may own the house.
Fortunately, many group LBS do not require precise locations  for each customer.
At last, we notice that the collision issue (i.e., two or more customers having the same KeyA/KeyB) will not happen in our proposed scheme since the certifying server will check the validity of each key in the login process.

\subsection{Customer Authentication}
With the help of the authentication center of cellular mobile service operators, the proposed IMSI pseudonym can prevent illegal access which will be blocked in the login process.
In the proposed IMSI pseudonym, a customer generates KeyA based on his/her own IMSI.
A customer can update location information only if he/she is an authorized customer.
Even if an adversary can steal a customer's SIM card to get IMSI, this adversary still cannot log in to the system without the customer's watchword.

%11-27
\subsection{Service Continuity}
In group LBS, service continuity is important because the system provider can record and analyze the historical location data of a customer.
The proposed IMSI-based pseudonym can retain the individual identification of a customer's pseudonyms because the unique IMSI is used in pseudonym generation.
Therefore, a customer can enquire his/her location record by providing KeyA to the certifying server.
In addition, historical location can be used for social data mining services.
In contrast, most of conventional random pseudonym approach will eliminate the individual identification of pseudonyms. Therefore, new location data cannot be appended to the same historic records when users change their identifications, reinstall the program, or change devices.
In this case, location records become difficult to be traced since the ownership information of a pseudonym is hard to be known.
With the uniqueness provided by IMSI, the proposed privacy protection scheme can be operated and retain the service continuity simultaneously.
%--------%
%One possible solution is to frequently update user identification in the location database.
%However, this solution is not practical since many cloud file systems, such as the Hadoop distributed file system (HDFS) \cite{gupta2014survey}, require a write-once-read-many access model.
%In other words, once a file is created, it cannot be modified.

\subsection{Comparison with Conventional Pseudonym Schemes}
Compared to conventional random pseudonym approaches, the proposed network coding based pseudonym has the following advantages.
First, as mentioned above, the proposed IMSI pseudonym can further improve data ownership privacy, customer authenticity, and service continuity.
Especially, the proposed pseudonym scheme can provide unconditional security rather than the computational security of Hash-based pseudonym.
The property of unconditional security offers great robustness against brute force attacks in cloud computing environment.
Besides, the designed two-tier coding method can reduce the computational complexity and energy consumption in mobile devices compared to Hash schemes.
A customer only needs to encode its identity one time when logging in to the system.
No further decoding procedures are required.
Moreover, the designed certifying server can authenticate and authorize individual customers.
Note that the certifying server handles the key collision and uncloaks the pseudonym only according to one's keyB.
Thus, the certifying process can be implemented in a distributed manner to avoid the service bottleneck and the single point of failure in the system \cite{zhang2011cross,li2011distributed}.
Last but not least, the proposed network coding based pseudonym is compatible with the existing pseudonym approaches and security communication protocols such as IP security (IPSec) and secure socket layer (SSL).

\section{Experimental Results}  \label{Experiment RESULTS}

\begin{figure}
\centering
\includegraphics[width=0.5\textwidth]{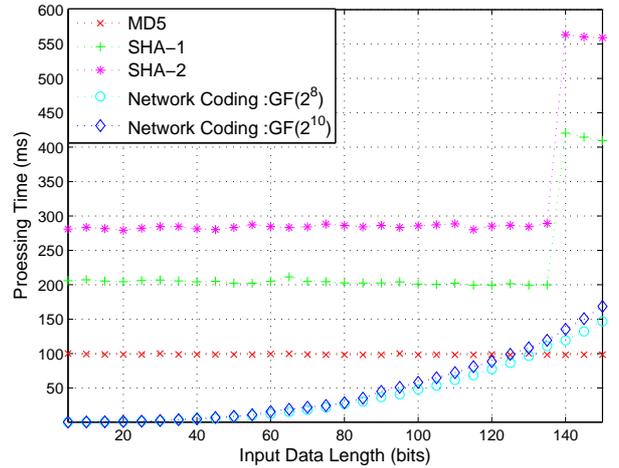}\\
\caption{Processing time for different key generation schemes.} \label{delay}
\end{figure}

Since mobile devices should perform key generation when customers log in to the LBS or update their locations, processing delay and power consumption for key generation function are important performance issues.
The applications of network coding may be limited because of the computational complexity and energy consumption in mobile devices \cite{lee2013understanding}.
To examine the delay performance and the energy consumption of the proposed network coding scheme, we implemented the JOIN services and performed experiments based on HTC Desire with a Qualcomm QSD8250, which is a low-power ARM based 1 GHZ processor. We choose the three most widely used Hash functions (i.e., MD5,  SHA-1, and SHA-2), and network coding schemes with different Galois field sizes.

\begin{figure}
\centering
\includegraphics[width=0.5\textwidth]{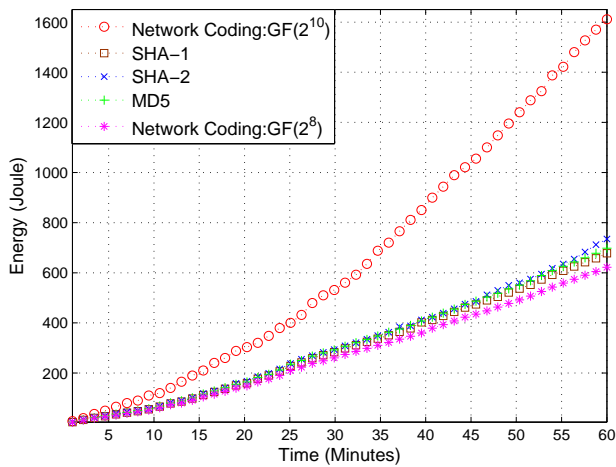}\\
\caption{Energy consumption of smartphones for different key generation schemes.} \label{power}
\end{figure}

Figure~\ref{delay} shows the processing time versus the length of IMSI for different key generation schemes. One can observe that the processing time of SHA increases sharply when the length of IMSI is around $140$ bits.
This is because input data are broken down into chunks for processing in SHA and MD5. Chunks of input data will be appended to fit in with the chunk size.
Also, the processing time of network coding approach is not highly affected by Galois field size because the Galois field mathematical operations are always in the range of field size rather than growing exponentially.
Finally, it is important to note that the processing time of network coding is shorten than that of Hash functions when the IMSI length is shorter than 120 bits. In the existing communication systems, a standard IMSI is usually 50 bits long or even shorter. Our results show that using network coding to encrypt a standard IMSI can reduce more than 95\% of processing time in mobile devices compared to using Hash functions.

Figure~\ref{power} shows the system energy consumption versus running time for different key generation schemes. The system energy, including both coding energy and transmission energy, is tested with eight packets coded together, each of one KB length, and a transmission energy per bit of 200 pJ/bit \cite{angelopoulos_11}. Paper \cite{angelopoulos_11} has identified when a small field size is used, the total system energy consumption is dominated by the extra RF retransmissions. In contrast, as field size becomes large, significantly increased energy is required for performing the coding process with less influence on the expected number of transmitted packets.
Thus, it is worthwhile to investigate the optimal field size of the proposed network coding scheme in terms of energy consumption. As shown in Fig.~\ref{power}, although network coding with field size ${2^{10}}$ consumes more energy  compared to MD5 or SHA, network coding with field size ${2^{8}}$ can reduce about 10\% of the energy consumption compared to MD5 or SHA. Hence, we can conclude that network coding with field size ${2^{8}}$ is a promising function for energy-limited devices.
Finally, it is worthwhile mentioning that the proposed scheme can provide unconditional security with any network coding parameters.

\section{Conclusions} \label{Conclusion}

In this paper, we proposed a lightweight network coding pseudonym scheme to protect the ownership of mobile data in an untrusted cloud database by disconnecting the relation of data and its identity.
We develop the IMSI-based group secure (IGS) algorithm for group LBS based on the proposed pseudonym scheme.
Our theoretical results show that the proposed scheme can achieve unconditionally secure privacy protection.
Our experimental results indicate that the proposed network coding based pseudonym cannot only exhibit better delay performance, but also provide lower energy consumption compared to Hash-based pseudonyms. Our proposed scheme is fully compatible with existing data security techniques. This work is the first step to exploit the potential of network coding in providing secure pseudonym. Besides the LBS, we will further investigate the applicability and limitation of network coding based pseudonym in the applications of machine-to-machine (M2M) communications.

\bibliography{Reference}
\bibliographystyle{IEEEtran}

\vspace{-3cm}
\begin{biography}[{\includegraphics[width=1in,height
=1.1in,clip,keepaspectratio]{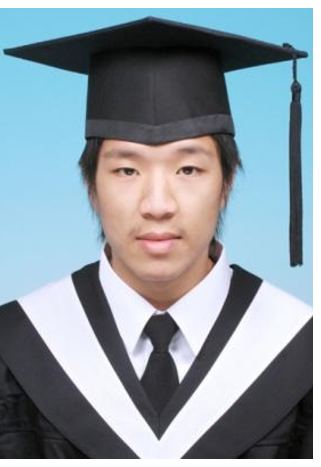}}]{Yu-Jia Chen}
received the B.S. degree and Ph.D. degree in electrical engineering from National Chiao Tung University, Taiwan, in 2010 and 2016, respectively. He is currently a postdoctoral fellow in National Chiao Tung University. His research interests include network coding for secure storage in cloud datacenters, software defined networks (SDN), and sensors-assisted applications for mobile cloud computing. Dr. Chen has published 15 conference papers and 3 journal papers. He is holding one US patent and two ROC patent.
\end{biography}

\vspace{-3cm}
\begin{biography}[{\includegraphics[width=1in,height
=1.1in,clip,keepaspectratio]{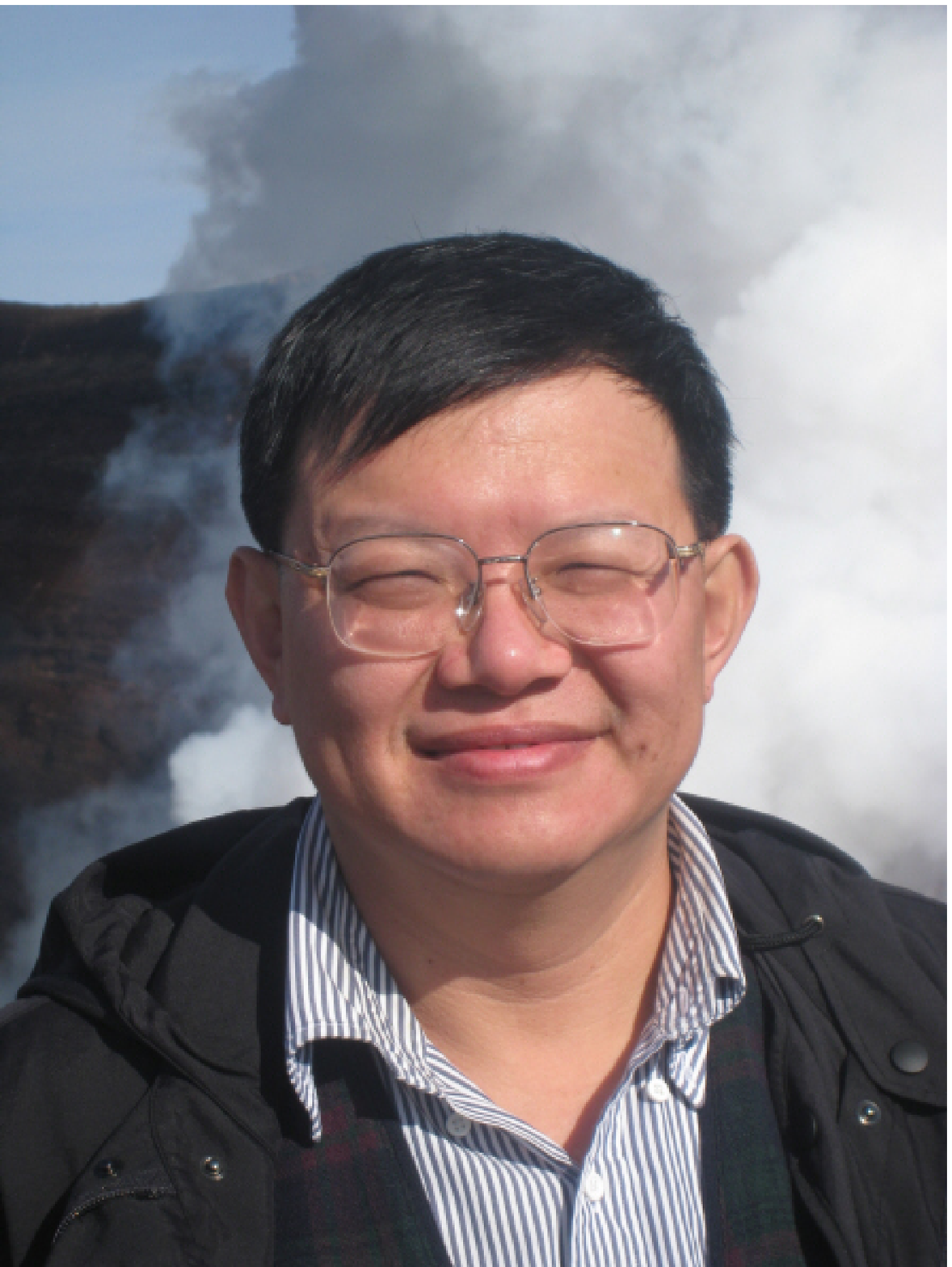}}]{Li-Chun Wang
(M'96 -- SM'06 -- F'11)} received the B.S. degree from National Chiao Tung University, Taiwan, R.O.C. in 1986, the M.S. degree from National Taiwan University in 1988, and the Ms. Sci. and Ph. D. degrees from the Georgia Institute of Technology , Atlanta, in 1995, and 1996, respectively, all in electrical engineering.

From 1990 to 1992, he was with the Telecommunications Laboratories of Chunghwa Telecom Co. In 1995, he was affiliated with Bell Northern Research of Northern Telecom, Inc., Richardson, TX. From 1996 to 2000, he was with AT\&T Laboratories, where he was a Senior Technical Staff Member in the Wireless Communications Research Department.  Since August 2000, he has joined the Department of Electrical and Computer Engineering of National Chiao Tung University in Taiwan and is the current Chairman of the same department. His current research interests are in the areas of radio resource management and cross-layer optimization techniques for wireless systems, heterogeneous wireless network design, and cloud computing for mobile applications.

Dr. Wang won the Distinguished Research Award of National Science Council, Taiwan in 2012, and was elected to the IEEE Fellow grade in 2011 for his contributions to cellular architectures and radio resource management in wireless networks. He was a co-recipient(with Gordon L. Stuber and Chin-Tau Lea) of the 1997 IEEE Jack Neubauer Best Paper Award for his paper ``Architecture Design, Frequency Planning, and Performance Analysis for a Microcell/Macrocell Overlaying System," IEEE Transactions on Vehicular Technology, vol. 46, no. 4, pp. 836-848, 1997. He has published over 200 journal and international conference papers. He served as an Associate Editor for the IEEE Trans. on Wireless Communications from 2001 to 2005, the Guest Editor of Special Issue on "Mobile Computing and Networking" for IEEE Journal on Selected Areas in Communications in 2005, "Radio Resource Management and Protocol Engineering in Future Broadband Networks" for IEEE Wireless Communications Magazine in 2006, and "Networking Challenges in Cloud Computing Systems and Applications," for IEEE Journal on Selected Areas in Communications in 2013, respectively. He is holding 10 US patents.
\end{biography}

\end{document}